\documentclass[12pt,draftcls,onecolumn]{IEEEtran}

\IEEEoverridecommandlockouts

\usepackage{times} 
\usepackage{color}
\usepackage{graphicx}
\usepackage{setspace}
\usepackage{bbm}
\usepackage{mathdots,mathrsfs}
\usepackage{amssymb,latexsym,amsfonts,amsmath,cite}
\usepackage{stmaryrd}
\usepackage{caption}
\usepackage[dvips]{psfrag}
\usepackage{setspace}
\usepackage{amsmath}
\usepackage{url}
\usepackage{pgf,tikz}

\newcommand{\R}{{\mathbb{R}}}

\newcommand{\ie}{{i.e., }}

\newcommand{\h}{{\text{H}}}

\newcommand{\HT}{\mathcal{H}_{2}}

\newcommand{\Tr}{\mathbf{Tr}}

\newcommand{\RR}{\boldsymbol{\rho}_{\textbf{ss}}}

\newtheorem{definition}{\bfseries Definition}

\newtheorem{remark}{\bfseries Remark}

\newtheorem{theorem}{\bfseries Theorem}

\newtheorem{corollary}{\bfseries Corollary}

\author{Milad Siami~and~Nader Motee
\thanks{M. Siami and N. Motee are with the Department of Mechanical Engineering and Mechanics, Packard Laboratory, Lehigh University, Bethlehem, PA, USA. Email addresses:  {\tt\small siami,motee@lehigh.edu}.
}}

\title{Scaling Laws for Disturbance Propagation in \\ Cyclic Dynamical Networks}
\begin{document}
\maketitle
\begin{abstract}
Our goal is to analyze performance  of  stable linear dynamical networks subject to external stochastic disturbances. The square of the $\mathcal H_2$-norm of the network is used as a performance measure to quantify the expected steady-state dispersion of the outputs of the network. We show that this performance measure can be tightly bounded from below and above  by some spectral functions of the state-space matrices of the network. This result is applied to a class of cyclic linear networks and shown that their performance measure scale quadratically with the network size. 
\end{abstract}


\allowdisplaybreaks

\section{Introduction}

The class of dynamical networks with cyclic interconnection topologies has been the focus of numerous studies over the past few years. Examples of dynamical networks with cyclic structures include gene regulation networks \cite{Tyson}, metabolic pathways \cite{chandra11, buzi11, Scardovi2010,jovarcsonTAC08}, and cellular signaling pathways \cite{Kholodenko}. The results of this paper have been motivated by cyclic dynamical networks arising in biological networks. In \cite{Tyson}, the authors propose a sufficient stability condition for an unperturbed cyclic network. Later on, these results were extended to show that global asymptotic stability conditions for cyclic dynamical networks can be obtained using diagonal stability and passivity properties of subsystems in the form of some secant conditions \cite{Arcak}. 
Although, stability properties of cyclic linear and nonlinear dynamical networks have been studied in several recent works, a comprehensive performance analysis is yet to be done for this class of networks.
In various applications such cyclic dynamical networks must operate in uncertain environments such as under influence of external stochastic disturbances. Therefore, one of the  relevant challenges is to study performance of cyclic dynamical networks under external stochastic disturbances. 

 
In this paper, we first present a result on performance analysis of general stable linear systems subject to external stochastic disturbances. We employe the square of the $\mathcal{H}_2$-norm of the system from disturbance input to the output  as a performance measure \cite{Bamieh12, Doyle89}. We derive explicit tight lower and upper bounds for this performance measure that are spectral functions of the state-space matrices of the system. Then, this general result is utilized to quantify inherent fundamental limits on the performance measure of a class of cyclic linear dynamical networks. This class of cyclic networks with asymmetric structures has been used to model certain biochemical pathways \cite{Kholodenko,Scardovi2010,jovarcsonTAC08}. We particularly show that the performance measure of a cyclic linear dynamical network scales quadratically with the network size. Moreover, it is shown that when all subsystems are  identical, the network attains the best achievable performance among all cyclic networks with the same secant criterion. 
 
%

The notation used in this paper is fairly standard. Specifically, $\R$ denotes the set of real numbers, $\mathbb C$ denotes the set of complex numbers, $\mathbf{Re}\{.\}$ denotes the real part of a complex number, $(.)^{\text T}$ denotes transpose and $(.)^{\text H}$ denotes Hermitian transpose. $I_{n \times n} \in \R^{n\times n}$ is the identity matrix, $\mathbf 0_n \in \R^n$ is a vector of all zeros. For a square matrix $A$, $\Tr(A)$ refers to the summation of on-diagonal elements of $A$.  We write $\lambda_{\max}(M)$ (resp., $\lambda_{\min}(M)$) for the maximum (resp., minimum) eigenvalue of the matrix $M$, $\mathbf{diag}[v]$ for a square diagonal matrix with the elements of vector $v$ and $\|.\|_2$ for 2-norm of a vector. The eigenvalues of a symmetric matrix $Q\in \R^{n \times n}$ are indexed in ascending order $\lambda_1(Q) \leq \lambda_2 (Q)\leq \cdots \leq \lambda_n(Q)$. $\mathbf E[v]$ stands for the expected value of random variable $v$.

\allowdisplaybreaks

\section{Problem Formulation}\label{sec:H-2-Norm}

	The steady-state variance of outputs of linear systems driven by external stochastic disturbances can be regarded as a measure of performance. We consider a linear time-invariant system 
	\begin{eqnarray}
		\dot{x} &=&  Ax~+~ \xi, \label{linear-1} \\ 
		y&=&Cx, \label{linear-2} 
	\end{eqnarray}
with $x(0)=\mathbf 0_n$, where $x \in \R^{n}$ is the state and $y \in \R^{m}$ the output of the system. The input signal $\xi \in \R^{n}$ is a white noise process with zero mean and identity covariance, i.e., 
	\begin{equation}
		\mathbf E \left [\xi(t) \xi(\tau)^{\text T} \right ] ~=~ I_{n\times n}\hspace{0.03cm} \delta(t - \tau),
		\label{covariance}
	\end{equation}
where $\delta(.)$ is the delta function. It is assumed that the state matrix $A$ is Hurwitz.   
 
\begin{definition}
	The $\HT$--norm of linear system \eqref{linear-1}-\eqref{linear-2} from $\xi$ to $y$ is defined as the square root of the following quantity 
	\begin{equation}
		\RR\left (A;{Q}\right) ~:=~  \lim_{t \rightarrow \infty} \mathbf E \left [ x^{\text T}(t)Qx(t)\right ], 
		\label{bound0} %
	\end{equation}
where $Q=C^{\text T}C$. 
\end{definition}

	For unstable linear systems, the outputs of the system have finite steady state variance as along as the unstable modes of the system are not observable from the output of the system (cf. \cite{Bamieh12}). 
The value of performance measure (\ref{bound0}) for \eqref{linear-1}-\eqref{linear-2} can be quantified as 
%
	\begin{equation}
		\RR\left (A;{Q}\right)~=~ \mathbf{Tr}( P ),
		\label{ricca3}
	\end{equation}
where $P$ is the unique solution of the Lyapunov equation
	\begin{equation}
		PA~+~A^{\text T}P~+~Q~=~0.
		\label{obser}
	\end{equation} 
%

Our first goal is to calculate tight lower and upper bounds on the performance measure (\ref{bound0}) for general stable linear time-invariant systems. Our second goal is to utilize these bounds to characterize size-dependent fundamental limits on the performance measure of cyclic linear dynamical networks.

\section{The Main Result}
\label{sec:main}
Since the performance measure (\ref{bound0}) is  real-valued, we show that this measure can be sandwiched between two real-valued functions of  eigenvalues of matrices $A$ and  $Q$.  

\begin{theorem}
\label{main}
In linear system  \eqref{linear-1}-\eqref{linear-2}, suppose that  the disturbance input is a white stochastic process $\xi$ with zero-mean and covariance  (\ref{covariance}),  the state matrix $A$ is Hurwitz, and $Q=C^{\text T}C$. Then, we have 
	\begin{equation} 
		-\sum_{i=1}^{n} \frac{\lambda_{\min}(Q)}{2\mathbf{Re}\{{\lambda_i}(A)\}} ~\leq ~\RR\left (A;Q\right)
~\leq~ -\sum_{i=1}^{n} \frac {\lambda_{i}(Q)}{2{\lambda_i}(A_s)}, \label{main-ineq}
	\end{equation}
where $A_s=\frac{A^{\text T}+A}{2}$ is the systematic part of matrix $A$ and Hurwitz. Moreover, the lower bound in (\ref{main-ineq}) is achieved if and only if $A$ is normal, \ie $A^{\text T}A=AA^{\text T}$, and $Q$ has $n$ identical eigenvalues. 
\end{theorem}

\begin{proof} Every symmetric matrix $Q$ can be decomposed as $Q=UDU^{\text T}$ where $UU^{\text T}=U^{\text T}U=I$ and $D=\mathbf{diag}\left [\lambda_1(Q), \cdots \lambda_n(Q)\right ]$. Using this fact, we can rewrite (\ref{obser}) in the following form 
%
	\begin{equation}
		\bar A^{\text T} \bar P~+~ \bar P\bar A~+~D~=~0,
		\label{lya}
	\end{equation}
where $\bar A= U^{\text T}AU$ and $\bar P= U^{\text T}PU $.
Since $A$ is Hurwitz, all its eigenvalues  have strictly negative real parts. Therefore, the unique solution of (\ref{lya}) can be expressed in the following closed form
	\begin{equation}
		\bar P~=~\int_{0}^{\infty} e^{\bar A^{\text T}t} D e^{\bar At}dt.
		\label{ppp}
	\end{equation}
According to Schur decomposition theorem \cite{hornJohnson90}, there exists a unitary matrix $V \in \mathbb C ^{n \times n}$ such that $\bar A=V(\Gamma+N)V^{\h}$ where $\Gamma=\mathbf{diag}\left[\lambda_1(A),\cdots,\lambda_n(A)\right]$, $N$ is strictly upper triangular, and $V^{\h}$ is the conjugate transpose of $V$.
Next, let us consider the integrand of (\ref{ppp})
	\begin{eqnarray}
		\mathbf{Tr}(e^{\bar A^{\text T}t} D e^{\bar At})&=&\mathbf{Tr}(e^{\bar A^{\h}t} D e^{\bar At}) \nonumber \\
		&=&\mathbf{Tr}(e^{(\Gamma^{\h}+N^{\h})t}V^{\h}DVe^{(\Gamma+N)t}V^{\h}V) \nonumber \\
		&=&\mathbf{Tr}(V^{\h}DVe^{(\Gamma^{\h}+N^{\h})t}e^{(\Gamma+N)t}) \nonumber \\
		&=&\mathbf{Tr}(DVe^{(\Gamma^{\h}+N^{\h})t}e^{(\Gamma+N)t}V^{\h})  \nonumber \\
		&\geq& \lambda_{\min}(Q) \mathbf{Tr}(Ve^{(\Gamma^{\h}+N^{\h})t}e^{(\Gamma+N)t}V^{\h}) \nonumber \\
		&=&\lambda_{\min}(Q) \mathbf{Tr}(e^{(\Gamma^{\h}+N^{\h})t}e^{(\Gamma+N)t}). 
		\label{pppp224}
	\end{eqnarray}
Furthermore, we have
	\begin{eqnarray}
		e^{(\Gamma+N)t} & = & e^{\Gamma t} +  M_t,\label{pppq111} \\
		e^{(\Gamma^{\h}+N^{\h})t} & = & e^{\Gamma^{\h}t} + M_t^{\h}, \label{pppq222}
	\end{eqnarray}
where  $M_t$ is an upper-triangular Nilpotent matrix. From \eqref{pppq111} and (\ref{pppq222}), we have
	\begin{eqnarray}
		\mathbf{Tr}(e^{(\Gamma^{\h}+N^{\h})t}e^{(\Gamma+N)t})& =& \mathbf{Tr}(e^{\Gamma t}e^{\Gamma^{\h}t} +  M_t M_t^{\h})  \nonumber \\
		&\geq& \mathbf{Tr}(e^{(\Gamma^{\h}+\Gamma)t}).
		\label{pppp222}
	\end{eqnarray}
From (\ref{pppp224}) and (\ref{pppp222}), it follows that 
	\begin{eqnarray}
		\mathbf{Tr}(e^{\bar A^{\text T}t} D e^{\bar At}) &\geq& \lambda_{\min}(Q) \mathbf{Tr}(Ve^{(\Gamma^{\h}+N^{\h})t}e^{(\Gamma+N)t}V^{\h}) 		\nonumber \\
		&\geq& \lambda_{\min}(Q) \mathbf{Tr}(e^{(\Gamma^{\h}+\Gamma)t}) \nonumber\\
		&=&\lambda_{\min}(Q) \mathbf{Tr}(e^{2\mathbf{Re}\{{\Gamma}\}t}).
		\label{p12ppp}
	\end{eqnarray}
Since $\mathbf{Re}\{\lambda_i(A)\} \neq 0$ for all $i=1,\ldots,n$, from (\ref{ppp}) and (\ref{p12ppp}) we have 
	\begin{eqnarray}
		\mathbf{Tr}(P)&=&\mathbf{Tr}(\bar P) \nonumber \\
		& = & \int_0^{\infty}\mathbf{Tr}(e^{\bar A^{\text T}t}De^{\bar At})dt \nonumber \\
		&\geq& -\sum_{i=1}^{n} ~\frac{\lambda_{\min}(Q)}{2 \mathbf{Re} \{ \lambda_i(A) \}}.
 		\label{ppp2}
	\end{eqnarray}
In the last inequality, we apply the fact that the trace and sum operators are linear and they can commute with the integral. Note that the lower bound is achieved if and only if equalities  in (\ref{p12ppp}) and (\ref{pppp222}) hold, or equivalently, when $Q$ has $n$ identical eigenvalues and $A$ is a normal matrix, \ie $A^{\text T}A=AA^{\text T}$.
In order to prove the RHS inequality in (\ref{main-ineq}), we use  \cite[Corollary $2.1.1$]{Komaroff92}, which gives us the upper bound in (\ref{main-ineq}).
\end{proof}

\begin{remark}
	If $Q=q I_{n \times n}$ for $q >0$, then the lower bound in Theorem \ref{main} is tighter than the lower bounds reported in reference papers  \cite{1103913,377,1103858,1104370}. 
\end{remark}

%

The following corollary explores one special case and shows how the performance measure (\ref{bound0}) depends on the general properties of $A$.

\begin{corollary}
\label{coro-1}
	Suppose that the assumptions of Theorem \ref{main} hold. Furthermore, if we assume that matrix $A$ is normal and matrix $Q$ has $n$  identical  eigenvalues, i.e., $\lambda(Q)=\lambda_i(Q)=\ldots=\lambda_n(Q)$, then  
	\begin{eqnarray} 
		\RR\left (A;Q \right)&=& -\sum_{i=1}^{n} \frac {\lambda(Q)}{2{\lambda_i}(A_s)},
		\label{rer}
	\end{eqnarray}
\end{corollary}
where $A_s=\frac{A^{\text T}+A}{2}$ is the systematic part of matrix $A$. 
\begin{proof}
	According to the Schur decomposition for normal matrices, there exists a unitary $V \in \mathbb C ^{n \times n}$, such that $A=V \Gamma V^{\text H}$ where $\Gamma= \mathbf{diag}\{\lambda_1(A),\ldots(A), \lambda_n(A)\}$ and $V^{\text H}$ denotes the conjugate transpose of matrix $V$. Using this we have
	\begin{eqnarray}
		A_s & = & \frac{A+A^{\text H}}{2}\nonumber \\
			& = & V\left(\frac{\Gamma+\Gamma^{\text H}}{2}\right)V^{\text H}\nonumber \\
			& = & V \mathbf{diag} \big(\mathbf{Re}\{\lambda_1(A)\},\ldots,\mathbf{Re}\{\lambda_n(A)\} \big)V^{\text H}.
		\label{normal-eq}
	\end{eqnarray}
This implies that $\lambda_{i}(A_{s})=\mathbf{Re}\{\lambda_i(A)\}$ for all $i=1,\ldots,n$. Thus, the lower and upper bounds in (\ref{main-ineq}) coincide and therefore, we get equality (\ref{rer}).
\end{proof}

All symmetric, skew-symmetric, and orthogonal matrices are examples of normal matrices. According to Corollary \ref{coro-1}, when $A$ is normal, the performance measure (\ref{bound0})  can be exactly equal to the upper bound and calculated as a function of eigenvalues of the symmetric part of $A$.  This should be differentiated from the result of Theorem \ref{main}, where it is shown that the lower bound is attained if and only if $A$ is normal and all eigenvalues of $Q$ are identical.

\section{Cyclic Linear Dynamical Networks}\label{sec:cyclic-network}
	\begin{figure}[t]
        \centering
        \begin{tikzpicture}
        \draw [ thick] (-2.5,0) circle [radius=0.4];
        \draw [ thick] (-1,0) circle [radius=0.4];
        \draw [ thick] (.5,0) circle [radius=0.4];
        \draw [  thick] (2.5,0) circle [radius=0.4];
        \draw [->, thick] (-2,0) -- (-1.5,0);
        \draw [->,  thick] (-.5,0) -- (0,0);
        \draw [fill] (1.3,0) circle [radius=.04];
        \draw [fill] (1.5,0) circle [radius=.04];
        \draw [fill] (1.7,0) circle [radius=.04];
        \draw [ thick, dashed] (3,0) -- (3.5,0);
        \draw [ thick, dashed] (3.5,0) -- (3.5,.8);
        \draw [ thick, dashed] (3.5,.8) -- (-3.5,.8);
        \draw [ thick, dashed] (-3.5,.8) -- (-3.5,0);
        \draw [ ->,thick, dashed] (-3.5,0) -- (-3,0);
        \node[] at (-2.5,0) {$x_1$};
        \node[] at (-1,0) {$x_2$};
        \node[] at (.5,0) {$x_3$};
        \node[] at (2.5,0) {$x_n$};
        \draw [ ->,thick] (-2.5,-.8) -- (-2.5,-.5);
        \draw [ ->,thick] (-1,-.8) -- (-1,-.5);
        \draw [ ->,thick] (.5,-.8) -- (.5,-.5);
        \draw [ ->,thick] (2.5,-.8) -- (2.5,-.5);
        \node[] at (-2.2,-.8) {$\xi_1$};
        \node[] at (-.7,-.8) {$\xi_2$};
        \node[] at (.8,-.8) {$\xi_3$};
        \node[] at (2.8,-.8) {$\xi_n$};
	   \end{tikzpicture}
  	   \caption{{\small Schematic diagram of negative feedback noisy cyclic system. The dashed link indicates a negative (inhibitory) feedback signal.}}
  	   \label{fig_2}
	\end{figure}
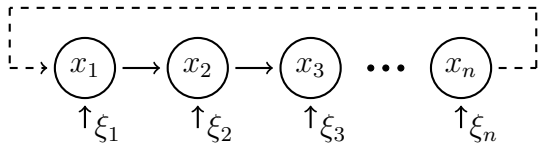
We apply our main result to a nontrivial and interesting class of linear systems with cyclic interconnection topologies. An example of a cyclic network is an autocatalytic pathway in biology with ring topology, which consists of a sequence of biochemical reactions where the system's product (output) is necessary to power and catalyze its own function  \cite{Siami13siam,Siami13acc}. We consider a cyclic linear dynamical network consists of a group of linear-time invariant systems $\mathcal S_i$ with state-space representations
	\begin{eqnarray}
	\mathcal S_i:\begin{cases}
		\dot x_i ~ = ~ -a_i x_i + u_i + \xi_{i}, \label{cyclic-1} \\
		v_i ~ = ~ c_ix_i, \label{cyclic-2}
		\end{cases}
	\end{eqnarray}
for $i=1,\ldots,n$, where $a_i$ and $c_i$ are strictly positive constants. The scalar quantities $u_i, v_i, x_i$ are the input, output and state variables of subsystem $\mathcal{S}_{i}$, respectively. By considering series interconnection of subsystems $\mathcal S_i$ for $i=1,2,\cdots,n$ and applying the output of subsystem $i$ as the input of subsystem $i+1$ (see Fig. \ref{fig_2}), we can represent the dynamics of the overall cyclic network as follows
	\begin{equation}
		\mathcal N: \begin{cases}
		\dot x_1 ~=~ -a_1 x_1 - v_n+\xi_1,\\ 
		\dot x_2 ~=~ -a_2 x_2 + v_1+\xi_2, \\
		~~~~~\vdots~ \\
		\dot x_n ~=~ -a_n x_n + v_{n-1}+\xi_n,\\
		\end{cases}	
		\label{noisy-system}
	\end{equation}
where $\xi_i$ for $i=1,2,\cdots,n$ are independent white stochastic processes with identical statistics. The resulting dynamical system takes the canonical compact form  \eqref{linear-1}-\eqref{linear-2} with state matrix 
	\begin{eqnarray}
		A=\left[\begin{array}{ccccc}
		-a_1 & 0    & \ldots & 0 & -c_n \\
		c_1  & -a_2 & \ldots & 0 & 0    \\
		\vdots & \vdots  & \ddots &\vdots  & \vdots   \\
		0 & 0 & \ldots & -a_{n-1} & 0 \\
		0 & 0 & \ldots & c_{n-1} & -a_{n} 
		\end{array}\right] \label{state-matrix}
	\end{eqnarray}
and output matrix $C=I_{n\times n}$. Our goal is to investigate robustness properties of the cyclic linear dynamical network (\ref{noisy-system}) driven by external white stochastic disturbances.  
\begin{theorem} \label{th-1}
	For the cyclic linear dynamical network  (\ref{noisy-system}) driven by a zero-mean white stochastic process with covariance (\ref{covariance}), let us define 
	\begin{eqnarray*}
		\mathfrak{a} & := & \sqrt[n]{a_1a_2 \cdots a_n}\label{constant-1}, \\
		\mathfrak{c} & := & \sqrt[n]{c_1c_2 \cdots c_n }\label{constant-2}.
	\end{eqnarray*}
If $\gamma > \cos(\frac{\pi}{n})$ where $\gamma=\frac{\mathfrak{a}}{\mathfrak{c}}$, then the cyclic linear dynamical network is stable. Moreover, if we assume that $\mathfrak{a}=a_1=\cdots=a_n$, then
	\begin{eqnarray}
		\RR\left (A;I_{n\times n}\right) &\geq&-\small{ \sum_{i=1}^{n} \frac{1}{2\mathbf{Re}\{{\lambda_i}(A)\}}}\nonumber \\
		&=&
		 \small{\left\{ \begin{array}{rcl}
     		\frac{n\tan \frac{\beta}{2}}{2 \mathfrak{c} \sin {\frac{\beta}{n}} }&,&\gamma < 1 \\
     		& & \\
     		\frac{n^2}{4 \mathfrak{c} }&,&\gamma = 1\\
      	& & \\
     		\frac{n\tanh \frac{{\beta}}{2}}{2 \mathfrak{c} \sinh {\frac{\beta}{n}} }&,&\gamma > 1 
		\label{eq7}
		\end{array} \right.}
		\label{h2}
	\end{eqnarray}    
where 
	\begin{eqnarray}
		\beta := \left\{ \begin{array}{rcl}
      	\mathrm{arcos}(\gamma) n&,&\gamma \leq 1 \\
     		\mathrm{arcosh}(\gamma) n&, &\gamma > 1
		\label{eq77}
		\end{array} \right..
		\label{beta}
	\end{eqnarray}   
{The equality in (\ref{h2}) is achieved if and only if $c_{1}=\ldots=c_{n}$, which means that all subsystems $\mathcal S_i$ for $i=1,\ldots,n$ are identical.}
\end{theorem}
\begin{proof}
	The stability condition $\gamma > \cos(\frac{\pi}{n})$ implies that $A$ is Hurwitz. Therefore, the $\mathcal H_2$--norm squared is finite and given by $\mathbf{Tr}(P)$ (see \cite{Doyle89} for more details), where $P$ is the unique positive definite solution of the Lyapunov equation 
	\begin{equation}
		AP~+~PA^{\text T}~=~-I_{n\times n}.
		\label{lya-1}
	\end{equation}
When $\mathfrak{a}=a_1=a_2=\cdots=a_n$, it is straightforward to verify that the characteristic equation of $A$ is given by 
\[(\lambda+\mathfrak{a})^n~+~c_1c_2 \cdots c_n~=~0.\] 
Therefore, the eigenvalues of the matrix are 
\[\lambda_{k}~=~-\mathfrak{a} \hspace{0.05cm}+\hspace{0.05cm}  \mathfrak{c} \hspace{0.05cm} e^{i(\frac{\pi}{n}+\frac{2\pi k}{n})},\] 
for $k=0,1,\cdots, n-1$. By substituting these eigenvalues into the lower bound of (\ref{main-ineq}), we get
	\begin{eqnarray}
		\hspace{-1cm}  -\sum_{i=1}^{n} \frac{1}{2\mathbf{Re}\{{\lambda_i}(A)\}}&=&  \sum_{k=0}^{n-1} ~\frac{1}{2 \mathbf{Re} \left\{-\mathfrak{a}+ 		\mathfrak{c} \hspace{0.05cm} e^{i(\frac{\pi}{n}+\frac{2\pi k}{n})} \right\}} \nonumber \\ 
		&=& \sum_{k=0}^{n-1} \frac{1}{2\mathfrak{c} \left( \gamma - \cos(\frac{\pi}{n}+\frac{2\pi k}{n}) \right )}.
		\label{eq23}
	\end{eqnarray}
First, let us assume that $\gamma<1$ and substitute  $\gamma=\text{cos}({\frac{\beta}{n}})$ in (\ref{eq23}). It follows that
	\begin{eqnarray}
		&& \hspace{-1.3cm} -\sum_{i=1}^{n} \frac{1}{2\mathbf{Re}\{{\lambda_i}(A)\}} ~=~   \frac{1}{2\mathfrak{c} }\sum_{k=0}^{n-1} \frac{1}{ \cos( {\frac{\beta}{n}}) - \cos(\frac{\pi}{n}+\frac{2\pi k}{n}) } \nonumber \\
		&& \hspace{0cm} = ~\frac{1}{4\mathfrak{c} }\sum_{k=0}^{n-1} \csc(\textstyle{\frac{(2k+1) \pi }{2n}+\frac{\beta}{2n}})\csc(\textstyle{\frac{(2k+1) \pi }{2n}-\frac{\beta}{2n}})\nonumber \\
		&&= ~\frac{n\tan \frac{\beta }{2}}{2 \mathfrak{c} \sin {\frac{\beta}{n}} },  \nonumber
	\end{eqnarray}
where the Birkhoff Ergodic theorem is used to conclude the last equation. Similar steps can be taken when $\gamma > 1$. In each case by substituting $\gamma$ from (\ref{eq77}) in (\ref{eq23}), one can obtain the desired result in the RHS of (\ref{h2}). According to Theorem \ref{main}, the equality in (\ref{h2}) is achieved if and only if $A$ is a normal matrix. On the other hand, based on the cyclic structure of matrix (\ref{state-matrix}) and the fact that $a_1=\ldots=a_n$, we conclude that $A$ is normal if and only if $c_1=\ldots=c_n$.
\end{proof}

	The classical secant criterion reported in \cite{Arcak} and \cite{Tyson } for cyclic linear dynamical network (\ref{noisy-system}) provides a stability condition when all $a_{i}$ for $i=1,\ldots,n$ are identical. This condition implies that the unperturbed system with $\xi =0$ in (\ref{noisy-system}) is stable if and only if $\gamma > \cos(\frac{\pi}{n})$. For a fixed parameter $\beta$, the stability condition of the cyclic network is not affected when   the number of intermediate subsystems changes. However, the result of Theorem \ref{th-1} asserts that the lower bound of the performance measure (\ref{bound0}) increases when the network size  increases.  
We show that the performance measure (\ref{bound0}) is in order of $\Omega(n^2)$ when parameter $\beta$ is fixed \footnote{We employ  the big omega notation in order to generalize the concept of asymptotic lower bound in the same
way as $\mathcal{O}$ generalizes the concept of asymptotic upper bound. We adopt the following definition according to \cite{Kunth76}:
\begin{equation} 
f(n) = \Omega(g(n)) ~\Leftrightarrow~g(n) = \mathcal{O}(f(n)), \label{big-omega}
\end{equation}
where $\mathcal{O}$ represents the big O notation. In the left hand side of \eqref{big-omega}, the $\Omega$ notation implies that $f(n)$ grows at least of the order of $g(n)$.}. More explicitly, we obtain the following approximation
	\begin{equation}
		- \sum_{i=1}^{n} \frac{1}{2\mathbf{Re}\{{\lambda_i}(A)\}} \approx \left\{ \begin{array}{rcl}
  		{\frac{\tan \frac{\beta}{2}}{{2\mathfrak{c} \beta }}}n^2&,&\gamma < 1 \\
		& & \\
    		{\frac{1}{{4 \mathfrak{c}}}}n^2&,&\gamma = 1\\
    		& & \\
    		{\frac{\tanh \frac{ { \beta}}{2}}{{2 \mathfrak{c} {\beta}}}}n^2&,&\gamma > 1 \\
		\label{eq777}
		\end{array} \right..
	\end{equation}
According to this calculation, we conclude that $\mathcal H_2$-norm of the cyclic network scales with $\Omega(n)$.  

\begin{remark}
\label{remark-2}
In Theorem \ref{th-1}, it is shown that performance measure of network (\ref{noisy-system}) is always greater or equal to the right hand-side of (\ref{h2}). When network $\mathcal N $ consists of $n$ identical subsystems, i.e. $\mathfrak{a} := a _{1}=\ldots=a_{n}$ and $\mathfrak{c} := c_{1}=\ldots=c_{n}$,  the performance measure attains its minimum value among all networks with the same parameters $\mathfrak{a}$ and $\mathfrak{c}$. Furthermore, its value is given by the RHS of (\ref{h2}).
\end{remark}

In the following corollary, we analyze the performance of a cyclic network when the output  of the network only depends on the state of the last subsystem, $x_i$ (see Figure \ref{fig_2}).

\begin{corollary}\label{coro-2}
	Suppose that the following condition holds for the cyclic linear dynamical network \eqref{noisy-system}
	\begin{equation}
		\frac{\mathfrak{a}}{\mathfrak{c}} ~>~ \cos\left(\frac{\pi}{n}\right), \label{stab-condition}
	\end{equation}
where $\mathfrak{a} := a _{1}=\ldots=a_{n}$, $\mathfrak{c} := c_{1}=\ldots=c_{n}$, and the output of the system is defined by 
	\begin{equation}
		y~=~Cx~=~\left[\begin{array}{cccc}0 & \ldots & 0 & 1\end{array}\right]x.
		\label{out2}
	\end{equation}
Then, the steady-state output dispersion is bounded from above by 
	\begin{eqnarray}
		\RR\left (A;{Q}\right) &:=& \lim_{t \rightarrow \infty} \mathbf{E}[y(t)^2] \nonumber \\
		&\leq& \frac{1}{2(\mathfrak{a}-\mathfrak{c}\cos(\frac{\pi}{n}))}.
		\label{coro2result}
	\end{eqnarray}
\end{corollary}
\begin{proof}
The steady-state output dispersion is given by  
	\[\RR\left (A;{Q}\right)~=~\mathbf{Tr}(CPC^{\text T}),\] 
where $P$ is the unique solution of the Lyapunov equation 
	\begin{equation}
		AP~+~PA^{\text T}~+~I_{n\times n}~=~0.
		\label{contro}
	\end{equation} 
According to Theorem \ref{th-1}, our assumption \eqref{stab-condition} implies that all the eigenvalues of $A$ have strictly negative real parts. Therefore, the unique solution of (\ref{contro}) can be written in the following closed form
	\begin{equation}
		P~=~\int_{0}^{\infty} e^{A^{\text T}t} e^{At}dt.
		\label{ppp10}
	\end{equation}
The state matrix defined by \eqref{state-matrix} is normal, i.e., $A^{\text T}A=AA^{\text T}$. According to the spectral theorem, there exists a unitary matrix $V \in \mathbb C ^{n \times n}$ such that $A=V \Lambda V^{\h}$, where $\Lambda=\mathbf{diag}\left[\lambda_1(A),\cdots,\lambda_n(A)\right]$. We now consider the integrand of (\ref{ppp10})
	\begin{eqnarray*}
		P&=&\int_{0}^{\infty} e^{A^{\h}t} e^{At} dt \nonumber \\
		&=&\int_{0}^{\infty}Ve^{\Lambda^{\h}t}e^{\Lambda t}V^{\h}dt \nonumber \\
		&=&V\mathbf{diag}\left(\frac{1}{2\mathbf{Re}\{\lambda_1(A)\}},\cdots,\frac{1}{2\mathbf{Re}\{\lambda_n(A)\}}\right)V^{\h}.
	\end{eqnarray*}
Since $\|C\|_{2}=1$, it follows that
	\begin{eqnarray}
		\mathbf{Tr}(CPC^{\text T})&\leq& \max_i \lambda_i(P)\nonumber \\
		&=&\max_i \frac{1}{2\mathbf{Re}\{\lambda_i(A)\}} \nonumber \\
		&=&\frac{1}{2(\mathfrak{a}-\mathfrak{c}\cos(\frac{\pi}{n}))}. 
		\label{g}
		\end{eqnarray}
Furthermore, the result of this corollary can be also considered as a direct consequence of Theorem \ref{main}. Since according to (\ref{out2}) matrix $Q$ has only one nonzero eigenvalue, $\lambda_n(Q)=1$, and the maximum eigenvalue of $A_s$ is given by
\[\lambda_n(A_s)~=~\mathfrak{a}-\mathfrak{c}\cos(\frac{\pi}{n}),\]
therefore, the desired result is obtained by using (\ref{main-ineq}).
		
	\end{proof}

The result of Corollary \ref{coro-2} is interesting as it implies that the performance measure scales by $\mathcal O(n^2)$, when parameter $\beta$ is fixed. This follows from rewriting the RHS of inequality (\ref{coro2result}) in the following form
	\begin{equation*}
		\frac{1}{2(\mathfrak{a}-\mathfrak{c}\cos(\frac{\pi}{n}))}~=~\frac{1}{4 \mathfrak{c} \sin(\frac{\pi -\beta}{n})\sin(\frac{\pi +\beta}{n})},
	\end{equation*}
where $\gamma \leq 1$, similar argument can be used where $\gamma  > 1$.

\section{Simulations}
\label{sec:simulation}
In order to verify our theoretical results, we consider cyclic network (\ref{noisy-system}) with $n$ identical subsystems $\mathcal S_i$ for $i=1,\ldots,n$. The asymptotic scaling of $\mathcal H_2$-norm for this network is depicted in terms of the total network size and parameter $\beta$ in Figure \ref{fig_222}. In this case, since subsystems are identical, the $\mathcal H_2$-norm of cyclic network (\ref{noisy-system}) can be calculated by the square root of the RHS of (\ref{h2}).  These values are depicted by small red circles (\textcolor{red}{$\circ$}) versus the number of subsystems ($n$). Moreover, these values are compared asymptotically to their approximation given by the square root of (\ref{eq777}). It can be observed that the square root of (\ref{eq777}) tightly approximates the $\mathcal H_2$-norm of cyclic network (\ref{noisy-system}). Based on Theorem \ref{th-1} even when subsystems $S_i$'s are not identical, we can say that Figure \ref{fig_222} portrays the lower bounds on the $\mathcal H_2$-norm of cyclic network (\ref{noisy-system}). This implies that $\mathcal H_2$-norm of the cyclic network scales with $\Omega(n)$ as long as parameter $\beta$ is fixed.  
	\begin{figure}[t]
		\begin{center}
		\psfrag{a}[c][c]{\small{Network size ($n$)}}
		\psfrag{c}[c][c]{\small{ $\mathcal H_2$ norm }}
		\psfrag{C}[c][c]{$x_3$}
		\psfrag{d}[c][c]{\tiny{Decreasing $\beta$ }}
		\psfrag{r}[c][c]{~~\tiny{where $\gamma<1$}}
		\psfrag{e}[c][c]{~~~~~~~\tiny{Decreasing $\beta$ }}
		\psfrag{f}[c][c]{~~~~~~~~~~\tiny{where $\gamma>1$}}
		\psfrag{E}[c][c]{$\xi$}
		\includegraphics[width=0.5\textwidth]{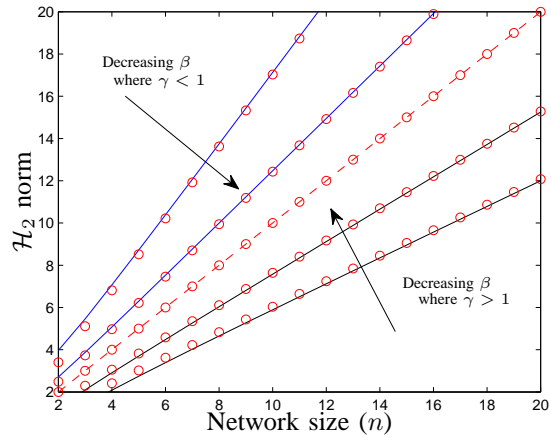}
		\end{center}
		\caption{{\small The lower bound in (\ref{noisy-system}), which is depicted by small red circles (\textcolor{red}{$\circ$}), is compared asymptotically to its approximation in (\ref{eq777}). It can be observed that (\ref{eq777}) tightly approximates the lower bound in (\ref{noisy-system}).}}
		\label{fig_222}
	\end{figure}

\section{Conclusion}
\label{sec:conclusion}
	In this paper, we have considered performance analysis of stable linear dynamical networks subject to external stochastic disturbances. We deployed the square of the $\mathcal H_2$--norm of the system from disturbance input to the output as a performance measure. Explicit formulae are derived for lower and upper bounds on the performance measure of the network. 
Then, we applied this result to an important class of cyclic linear networks and demonstrated that their performance measure scale quadratically with the network size. 
Finally, it is shown that when all subsystems are identical, the network attains the best achievable performance among all cyclic networks with the same secant criterion.

\begin{spacing}{}
\bibliographystyle{IEEEtran}
\bibliography{IEEEabrv,references-Milad-2014}
\end{spacing}

\end{document}